\documentclass[12pt,a4paper]{article}

\usepackage{amsmath,amssymb}
\usepackage{mathtools}
\usepackage{graphicx}
\usepackage[nosort]{cite}

\usepackage{tikz}
\usetikzlibrary{decorations.markings}
\usepackage{epsfig}  
\usepackage{enumitem}
\usepackage{relsize}
\usepackage{dsfont}
\usetikzlibrary{calc} 
\usetikzlibrary{patterns} 

\newcommand{\mathsym}[1]{{}}

\usepackage{amsthm,amsbsy,color,multicol}
\usepackage{array,calc}
\usepackage{rotating}
\usepackage{a4wide,bm}
\usepackage{bbm}
\usepackage[a4paper,text={450pt,650pt},centering]{geometry}
\usepackage{setspace}
\usepackage{marginnote}

\let\oldbfseries=\bfseries
\let\oldmdseries=\mdseries
\let\oldnormalfont=\normalfont
\renewcommand{\bfseries}{\oldbfseries\boldmath}
\renewcommand{\mdseries}{\oldmdseries\unboldmath}
\renewcommand{\normalfont}{\oldnormalfont\unboldmath}

\allowdisplaybreaks[3]

\numberwithin{equation}{section}

\makeatletter
\renewcommand\subparagraph{\@startsection{subparagraph}{5}%
	{\parindent}
	{0pt}
	{-1em}
	{\normalfont\itshape}
} 
\makeatother  

\usepackage[font=small,labelfont=bf,width=0.85\textwidth]{caption}

\ifx\hypersetup\sadfkjashdfkxja\newcommand\hypersetup[1]{}\fi

\hypersetup{plainpages=false}
\hypersetup{pdfpagemode=UseNone}
\hypersetup{bookmarksnumbered=true}
\hypersetup{pdfstartview=FitH}
\hypersetup{colorlinks=false}
\hypersetup{citebordercolor={.5 1 .5}}
\hypersetup{urlbordercolor={.5 1 1}}
\hypersetup{linkbordercolor={1 .7 .7}}

\DeclareMathSymbol{\Gamma}{\mathalpha}{letters}{"00}
\DeclareMathSymbol{\Delta}{\mathalpha}{letters}{"01}
\DeclareMathSymbol{\Theta}{\mathalpha}{letters}{"02}
\DeclareMathSymbol{\Lambda}{\mathalpha}{letters}{"03}
\DeclareMathSymbol{\Xi}{\mathalpha}{letters}{"04}
\DeclareMathSymbol{\Pi}{\mathalpha}{letters}{"05}
\DeclareMathSymbol{\Sigma}{\mathalpha}{letters}{"06}
\DeclareMathSymbol{\Upsilon}{\mathalpha}{letters}{"07}
\DeclareMathSymbol{\Phi}{\mathalpha}{letters}{"08}
\DeclareMathSymbol{\Psi}{\mathalpha}{letters}{"09}
\DeclareMathSymbol{\Omega}{\mathalpha}{letters}{"0A}

\DeclareMathOperator{\arctanh}{arctanh}

\newcommand{\gen}[1]{\mathrm{#1}}

\newcommand{\dd}{\mathrm{d}}
\newcommand{\ii}{\mathrm{i}}

\newcommand*\widebar[1]{%
  \hbox{%
    \vbox{%
      \hrule height 0.5pt 
      \kern0.25ex
      \hbox{%
        \kern-0.3em
        \ensuremath{#1}%
        \kern-0.1em
      }%
    }%
  }%
}

\ifx\genfrac\sdflkaj\else\fi

\newcommand{\alg}[1]{\mathfrak{#1}}

\newcommand{\beq}{\begin{equation}}
\newcommand{\eeq}{\end{equation}}

\def\[{\begin{equation}}
\def\]{\end{equation}}
\def\<{\begin{eqnarray}}
\def\>{\end{eqnarray}}

\newtheorem{mydef}{Definition}

\newtheorem{theorem}{Theorem}
\newtheorem{lemma}{Lemma} 
\newtheorem{remark}{Remark}
\newtheorem{proposition}{Proposition}
\newtheorem{corollary}{Corollary}
\newtheorem{example}{Example}

\makeatletter
\def\mr@ignsp#1 {\ifx\:#1\@empty\else #1\expandafter\mr@ignsp\fi}%
\newcommand{\multiref}[1]{\begingroup
\xdef\mr@no@sparg{\expandafter\mr@ignsp#1 \: }%
\def\mr@comma{}%
\@for\mr@refs:=\mr@no@sparg\do{\mr@comma\def\mr@comma{,}\ref{\mr@refs}}%
\endgroup}
\makeatother

\newcommand{\hypref}[2]{\ifx\href\asklfhas #2\else\href{#1}{#2}\fi}
\newcommand{\Secref}[1]{Section~\multiref{#1}}

\renewcommand{\eqref}[1]{(\multiref{#1})}

\makeatletter
\newlength{\apb@width}
\newcommand{\autoparbox}[2][c]{\settowidth{\apb@width}{#2}\parbox[#1]{\apb@width}{#2}}

\makeatother

\ifx\href\asklfhas\newcommand{\href}[2]{#2}\fi

\begin{document}

\renewcommand{\thefootnote}{\fnsymbol{footnote}}
\thispagestyle{empty}
\begin{flushright}\footnotesize
ZMP-HH/15-14 \\
ITP-UU-15/06
\end{flushright}
\vspace{0.2cm}

\begin{center}%
{\Large\bfseries%
\hypersetup{pdftitle={Differential approach to on-shell scalar products in six-vertex models}}%
Differential approach  \\ to on-shell scalar products \\ in six-vertex models%
\par} \vspace{2cm}%

\textsc{W.~Galleas$^a$ and J.~Lamers$^b$}\vspace{5mm}%
\hypersetup{pdfauthor={Wellington Galleas}}%

$^a$ \textit{II. Institut f\"ur Theoretische Physik \\ Universit\"at Hamburg, Luruper Chaussee 149 \\ 22761 Hamburg, Germany}\vspace{3mm}%

\verb+wellington.galleas@desy.de+ %

\vspace{0.5cm}

$^b$ \textit{Institute for Theoretical Physics, \\ Center for Extreme Matter and Emergent Phenomena, Utrecht University  \\ Leuvenlaan 4, 3584 CE Utrecht, the Netherlands}\vspace{3mm}%

\verb+j.lamers@uu.nl+

\par\vspace{2cm}

\textbf{Abstract}\vspace{7mm}

\begin{minipage}{12.7cm}
In this work we obtain hierarchies of partial differential equations describing on-shell scalar products for two types of 
six-vertex models. More precisely, six-vertex models with two different diagonal boundary conditions are considered: the case with
boundary twists and the case with open boundary conditions. Solutions and properties of our partial differential equations 
are also discussed.

\hypersetup{pdfkeywords={Six-vertex model, scalar product, partial differential equations}}%
\hypersetup{pdfsubject={}}%

\end{minipage}
\vskip 1cm
{\small May 2015}

\end{center}

\newpage
\renewcommand{\thefootnote}{\arabic{footnote}}
\setcounter{footnote}{0}

\tableofcontents

\section{Introduction}
\label{sec:intro}

The behavior of physical systems at criticality is endowed with remarkable properties and its description is intimately associated
with the study of correlation functions. When a system is away from a critical point, two-point correlation functions are expected
to decay exponentially as the distance between the target points becomes infinitely large. 
This asymptotic behavior motivates the definition of the \textit{correlation length} which contains fundamental information concerning
the system critical behavior. For instance, correlation lengths are expected to diverge in second-order phase transitions when a system 
reaches its critical point. In fact, this singular behavior can be regarded as a trademark of a second-order phase transition and this divergence
is then governed by a power law characterized by a \textit{critical exponent}. 

The scenario at criticality is drastically different and correlation functions are expected to decay according to a power law instead of
exponentially. This feature motivates one to define the \textit{anomalous dimension} for the corresponding order parameter.
The anomalous dimension is also a critical exponent and this particular power-law behavior shows that statistical fluctuations of the order
parameter are strongly correlated throughout the entire extension of the physical system.
In the vicinity of a phase transition  correlation lengths are much larger than the system's lattice spacing and 
it can be regarded as a natural distance scale. Furthermore, the correlation length is a function of the model's coupling constants
and the configurations of the system over the correlation length can be made sufficiently smooth such that the continuum formalism of 
quantum field theory is appropriate. For a detailed discussion we refer the reader to \cite{Baxter_book, Ma_book, Justin_book}. 

The computation of correlation functions of interacting systems is a highly nontrivial task  and it is often performed by means of
perturbative expansions. However, since critical behavior is associated with divergences of the aforementioned quantities,
it is highly desirable having a non-perturbative description of correlation functions. In this way, exactly solvable models of statistical
mechanics and quantum field theory are natural candidates for the investigation of critical phenomena.

In the context of quantum field theory correlation functions  are usually described by means of \textit{differential equations}.
We have for instance the Callan-Symanzik equation arising from the renormalization group framework \cite{Callan_1970, Symanzik_1970, Symanzik_1971} and the 
Knizhnik-Zamolodchikov (KZ) equation describing correlation functions of primary fields in conformal field theories \cite{Knizhnik_1984}. 
Differential equations also seem to play a prominent role in the study of correlation functions of exactly solvable models of statistical
mechanics, although results in this direction are still limited to a few models. For instance, in \cite{Its_1990} the authors have 
shown that a two-point function of the impenetrable Bose gas can be described by a differential equation of Painlev\'e type. 
Prior to that, Painlev\'e equations were also found in the study of certain two-point functions of the two-dimensional Ising 
model \cite{Barouch_1973, Tracy_1973, Wu_1976, Jimbo_1980}.

Within the framework of the Quantum Inverse Scattering Method (QISM) \cite{Sk_Faddeev_1979, Takh_Faddeev_1979}, the study of form factors and
correlation functions is intimately related to the evaluation of scalar products of Bethe vectors \cite{Korepin_1982}. In particular, as far 
as the six-vertex model and the associated Heisenberg spin chain are concerned, the form factor of the particle-number operator has been written
in terms of on-shell scalar products of Bethe vectors in \cite{Slavnov_1989}. It is worth remarking that the relation between form factors and
on-shell scalar products obtained in \cite{Slavnov_1989} builds on a detailed analysis of form factors previously investigated in 
\cite{Izergin_1984, Izergin_1985}. 

The above discussion raises some questions concerning the description of correlation functions through differential equations. For instance, 
there is a fundamental difference between the differential equations satisfied by correlation functions in conformal field theories 
and those obtained for the Bose gas and the Ising model. The latter are non-linear equations of Painlev\'e type whereas KZ 
equations are linear. Hence one may wonder if there exists \textit{linear} partial differential equations (PDEs) describing six-vertex model correlation functions.
This question is the main motivation for the present paper. Our arguments here shall rely on the results presented in \cite{Slavnov_1989}
expressing form factors of the six-vertex model in terms of on-shell scalar products of Bethe vectors. In this way we achieve a linear
differential description of six-vertex model form factors by presenting linear PDEs determining on-shell scalar products.   

The possibility of deriving linear PDEs for such quantities was first put forward in \cite{Galleas_2011} based on a detailed analysis of certain functional equations 
originated from the Yang-Baxter algebra. Several quantities and lattice systems have been tackled through this \textit{algebraic-functional method} and we refer the reader
to \cite{Galleas_2008, Galleas_2010, Galleas_2011, Galleas_2012, Galleas_SCP, Galleas_proc, Galleas_Lamers_2014, Galleas_openSCP} for an account of the results
so far. In particular, linear PDEs describing partition functions with domain-wall boundaries and spectral problems for transfer matrices have been 
presented in \cite{Galleas_2011, Galleas_proc, Galleas_Lamers_2014} and \cite{Galleas_2015} respectively. 
Although the Yang-Baxter algebra is a fundamental ingredient, this approach is not limited to this particular algebraic structure. For example, integrable systems with open boundary
conditions can also be included in this framework by replacing the Yang-Baxter algebra by the \textit{reflection algebra} \cite{Galleas_Lamers_2014, Galleas_openSCP}.

\paragraph{Outline.} In the present work we construct a differential description of six-vertex models on-shell scalar products for two types of boundary conditions.
This paper is organized as follows. \Secref{sec:twist} is concerned with the case of diagonal boundary conditons, also referred to
as quasi-periodic boundaries. Our analysis in this case will rely heavily on results previously presented in
\cite{Galleas_SCP} which includes a functional equation for the relevant scalar product. These results are recalled in \Secref{sec:funEQ} and the corresponding PDEs are
obtained and studied in \Secref{sec:PDE}. On-shell scalar products for diagonal open boundary conditions are discussed in \Secref{sec:open}. Here we make use of the results recently
obtained in \cite{Galleas_openSCP} whose details are given in \Secref{sec:funEQ_open}. The resulting PDEs are described in \Secref{sec:PDE_open} and we conclude with final 
remarks in \Secref{sec:CONCLUSION}.

\section{Twisted boundary conditions}
\label{sec:twist}

The study of integrability-preserving boundary conditions is an important branch of the theory of exactly solvable models \cite{Baxter_book}.
For instance, it is well known that particular boundary terms can be used to relate the critical behavior of several models of statistical mechanics 
\cite{Baxter_book, Alcaraz_Barber_1987, Alcaraz_1987}. Relations between different systems away from the criticality are more subtle but nevertheless such relations 
still exist. This is the case for the six-vertex models with domain-wall and anti-periodic boundary conditions \cite{Galleas_Twists}, whose relation has been established
through the same algebraic-functional approach employed in \cite{Galleas_SCP}. As far as the six-vertex model is concerned, boundary conditions are not a mere detail that 
can be disregarded in the thermodynamical limit. In that case boundary conditions even influence the model bulk free-energy 
\cite{Korepin_Justin_2000, Korepin_Ribeiro_2015, Bleher_2006, Bleher_2009, Bleher_2010} and modify the spectrum  of finite-size corrections \cite{Cardy_1986, Alcaraz_1987}. The latter provides fundamental data concerning the underlying conformal field theory. 
We also point out that twisted boundary conditions arise naturally in the context of Bethe/gauge correspondence by including certain topological terms
in the lagrangian of the two-dimensional $\mathcal{N} = (2,2)$ gauge theory \cite{Nekrasov_2009}.

As far as vertex models with twisted boundary conditions are concerned, the algebraic formulation of the integrability-preserving case
was put forward in \cite{deVega_1984} based on the invariance of the Yang-Baxter algebra under a particular linear map. This algebraic 
setting was then used in \cite{Galleas_SCP} to construct a functional equation governing on-shell scalar products. We shall build our analysis
on those results in order to derive a hierarchy of PDEs describing such scalar products.

\subsection{Functional equation}
\label{sec:funEQ}

The starting point of our analysis is the study of six-vertex model scalar products by means of functional equations as described in \cite{Galleas_SCP}.
Throughout this work we shall use the terminology \textit{on-shell scalar product} to refer to a scalar product of Bethe vectors 
when only \textit{one} of the vectors is subjected to Bethe ansatz equations. We consider these on-shell Bethe roots as fixed complex
parameters, so that the on-shell scalar products can be regarded as a function $\mathcal{S}_n \colon \mathbb{C}^n \to \mathbb{C}$
where $n \in \mathbb{Z}_{\geq 0}$ is the number of magnons contained in the Bethe vectors. In its turn, the integer index $n$ is bounded by the size
of the rectangular lattice on which our six-vertex model is defined. In fact, if $L \in \mathbb{Z}_{> 0}$ is the lattice horizontal length, then $0 \leq n \leq L$.
The on-shell scalar product $\mathcal{S}_n$ is a multivariate function in $n$ complex variables referred to as spectral parameters. Before proceeding let us pause for a moment
to set up our notation. 

\begin{mydef}[Indices, variables and parameters] \label{defvar}
Let $n \in \mathbb{Z}_{\geq 0}$ satisfy the condition $0 \leq n \leq L$ for $L \in \mathbb{Z}_{> 0}$. Then consider indices $k \in \mathbb{Z}$ and
let $\lambda_k, \lambda_k^{B} \in \mathbb{C}$ denote spectral variables. In addition to that, we shall use $\gamma \in \mathbb{C}$ to denote the anisotropy 
parameter, $\mu_k \in \mathbb{C}$ will be inhomogeneity parameters and $\phi_1, \phi_2 \in \mathbb{C} \backslash\{0 \}$ will characterize the boundary twists. 
For convenience, we shall also employ the short-hand notation: $a(\lambda) \coloneqq \sinh(\lambda + \gamma)$, $b(\lambda) \coloneqq \sinh(\lambda)$
and $c(\lambda) \coloneqq \sinh(\gamma)$.
\end{mydef}  

\begin{mydef}[Spectral parameter and inhomogeneity sets] \label{Uset}
Define $\gen{\Lambda}^{i,j} \coloneqq \{ \lambda_k \in \mathbb{C} \; | \; i \leq k \leq j \}$ and $\gen{\Lambda}^{i,j}_{\lambda} \coloneqq \gen{\Lambda}^{i,j} \backslash \{ \lambda \}$. 
It is also convenient to define the inhomogeneity set $\mathcal{U} \coloneqq \{ \mu_k \in \mathbb{C} \; | \; 1 \leq k \leq L \}$ in order to recast the dependence of the six-vertex model scalar products with the inhomogeneity parameters.
\end{mydef}

\begin{mydef}[Bethe roots]  \label{bethe}
Let $\lambda_k^{B} \in \mathbb{C}$ for $1 \leq k \leq n$. We then define the set $\mathcal{B} \coloneqq \left\{ \lambda_1^{B}, \lambda_2^{B}, \dots , \lambda_n^{B} \right\}$
of spectral variables solving
\<
\label{setB}
\prod_{\mu \in \mathcal{U}} \frac{a(\lambda_k^B - \mu)}{b(\lambda_k^B - \mu)} = (-1)^{n-1} \frac{\phi_2}{\phi_1}  \prod_{\substack{i=1 \\ i \neq k}}^n \frac{a(\lambda_k^B - \lambda_i^B)}{a(\lambda_i^B - \lambda_k^B)}  \; .  
\>
The set $\mathcal{B}$ will also be referred to as Bethe roots set.
\end{mydef}

\begin{remark}
The constraint imposed on the variables $\lambda_k^{B}$, i.e. Bethe ansatz equations, admits several inequivalent sets of Bethe
roots. Througout this work we let $\mathcal{B}$ be any single fixed set of such Bethe roots. 
\end{remark}

Among the main results of \cite{Galleas_SCP} we have the following theorem.

\begin{theorem}[Functional equation] \label{funEQ}
Taking into account the above definitions, the on-shell scalar product $\mathcal{S}_n$ satisfies the functional equation
\[
\label{FS}
K_0 \; \mathcal{S}_n (\gen{\Lambda}^{1,n}) + \sum_{\lambda \in \gen{\Lambda}^{1,n}} K_{\lambda} \; \mathcal{S}_n (\gen{\Lambda}_{\lambda}^{0,n}) = 0 \; ,
\]
with coefficients $K_0$ and $K_{\lambda}$ given by
\<
\label{coeff}
K_0 &\coloneqq& \phi_1 \prod_{\mu \in \mathcal{U}} a(\lambda_0 - \mu) \left[ \prod_{\lambda \in \gen{\Lambda}^{1,n}} \frac{a(\lambda - \lambda_0)}{b(\lambda - \lambda_0)} -  \prod_{\lambda \in \mathcal{B}} \frac{a(\lambda - \lambda_0)}{b(\lambda - \lambda_0)}  \right] \nonumber \\
&& + \; \phi_2 \prod_{\mu \in \mathcal{U}} b(\lambda_0 - \mu) \left[ \prod_{\lambda \in \gen{\Lambda}^{1,n}} \frac{a(\lambda_0 - \lambda)}{b(\lambda_0 - \lambda)} -  \prod_{\lambda \in \mathcal{B}} \frac{a(\lambda_0 - \lambda)}{b(\lambda_0 - \lambda)}  \right] \nonumber \\
K_{\lambda} &\coloneqq& \phi_1 \frac{c(\lambda_0 - \lambda)}{b(\lambda_0 - \lambda)} \prod_{\mu \in \mathcal{U}} a(\lambda - \mu) \prod_{\tilde{\lambda} \in \gen{\Lambda}^{1,n}_{\lambda}} \frac{a(\tilde{\lambda} - \lambda)}{b(\tilde{\lambda} - \lambda)} \nonumber \\
&& + \; \phi_2 \frac{c(\lambda - \lambda_0)}{b(\lambda - \lambda_0)} \prod_{\mu \in \mathcal{U}} b(\lambda - \mu) \prod_{\tilde{\lambda} \in \gen{\Lambda}^{1,n}_{\lambda}} \frac{a(\lambda - \tilde{\lambda})}{b(\lambda - \tilde{\lambda})} \; .
\> 
\end{theorem}
\noindent We refer the reader to \cite{Galleas_SCP} for a detailed proof.

Some comments on Theorem \ref{funEQ} are important at this point. For instance, Eq. (\ref{FS}) is obtained as a particular linear combination
of two equations describing more general six-vertex model scalar products, i.e. off-shell scalar products. In that case the scalar product  
depends on two sets of variables, i.e. $\mathcal{S}^{\text{off shell}}_n = \mathcal{S}^{\text{off shell}}_n (\lambda_1 , \dots , \lambda_n | \lambda_1^B , \dots , \lambda_n^B )$,
and the on-shell case described by (\ref{FS}) results from the restriction to fixed $\lambda_k^B \in \mathcal{B}$. 

Interestingly, Eq. (\ref{FS}) exhibits the same linear structure as the functional equations satisfied by partition functions of six-vertex
models with domain-wall boundaries \cite{Galleas_2013, Galleas_proc, Galleas_Lamers_2014}. The explicit form of the coefficients (\ref{coeff})
are the main difference between (\ref{FS}) and the equations derived in \cite{Galleas_2013, Galleas_proc, Galleas_Lamers_2014}.
In this way one can expect that some features discussed in \cite{Galleas_proc} will also hold for (\ref{FS}). 
Before investigating this possibility we present some further properties of on-shell scalar products that will be relevant.

\subsubsection{Operatorial formulation}
\label{sec:OPD}

While on-shell scalar products are functions in $n$ variables, we can readily see that Eq. (\ref{FS})
runs over the set $\gen{\Lambda}^{0,n}$ containing $n+1$ variables. Thus we have one extra variable, namely $\lambda_0$,
which will play a fundamental role for extracting a hierarchy of PDEs from (\ref{FS}). In fact, this differentiated role played 
by $\lambda_0$ can be exploited to achieve an operatorial formulation of (\ref{FS}) with the help of the following definition.

\begin{mydef} \label{dia}
Write $\mathbb{C}[\lambda_1^{\pm 1} , \lambda_2^{\pm 1} , \dots , \lambda_n^{\pm 1}]$ for the set of meromorphic functions in $n$ variables.
For each $\lambda \in \gen{\Lambda}^{1,n}$ and any $\bar{\lambda} \notin \gen{\Lambda}^{1,n}$ define the operator
\[
D_{\lambda}^{\bar{\lambda}} \colon \mathbb{C}[\lambda_1^{\pm 1} , \dots , \lambda^{\pm 1} , \dots , \lambda_n^{\pm 1}] \to \mathbb{C}[\lambda_1^{\pm 1} , \dots , \bar{\lambda}^{\pm 1} , \dots , \lambda_n^{\pm 1}] 
\]
acting on $f \in \mathbb{C}[\lambda_1^{\pm 1} , \lambda_2^{\pm 1} , \dots , \lambda_n^{\pm 1}]$ as
\[
\label{DIA}
\left( D_{\lambda}^{\bar{\lambda}} \; f \right) (\lambda_1, \dots , \lambda, \dots , \lambda_n) \coloneqq f(\lambda_1, \dots , \bar{\lambda}, \dots , \lambda_n) \; .
\]
\end{mydef}

This operator was firstly introduced in \cite{Galleas_2011} and it fits naturally in the structure of (\ref{FS}).
Using this operator we can rewrite our functional equation as $\mathfrak{L}(\lambda_0) \; \mathcal{S}_n (\gen{\Lambda}^{1,n}) = 0$ where 
$\mathfrak{L}$ is an operator given in terms of $D_{\lambda}^{\lambda_0}$. More precisely, this operator reads
\[
\label{Lop}
\mathfrak{L}(\lambda_0) \coloneqq K_0 + \sum_{\lambda \in \gen{\Lambda}^{1,n}} K_{\lambda} \; D_{\lambda}^{\lambda_0} \; .
\]

The importance of this reformulation goes beyond just an optimized notation. For instance, this operatorial prescription localizes the whole
dependence of (\ref{FS}) on $\lambda_0$ into the operator $\mathfrak{L}$. This feature does not depend on particular realizations of the 
operator $D_{\lambda}^{\lambda_0}$, which can suitably chosen according to the function space where the desired scalar product sits.
Interestingly, the operator $D_{\lambda}^{\lambda_0}$ admits a differential realization with only finitely many derivatives in the ring
of multivariate polynomials \cite{Galleas_2011}. This feature will be a key ingredient for our differential description of on-shell scalar products to be discussed in 
\Secref{sec:PDE} and \Secref{sec:PDE_open}.

\subsubsection{Properties of $\mathcal{S}_n$}
\label{sec:properties}

Theorem \ref{funEQ} has its roots in the Yang-Baxter algebra and carries very little information on the particular
representation of Bethe vectors we employ to define the scalar product $\mathcal{S}_n$. Therefore, one can expect
that there will exist different classes of solutions corresponding to different representations of Bethe vectors.
Linearity is one of the main properties of (\ref{FS}) and due to that combinations of solutions are also solutions. 
Thus we need to provide conditions capable of singling out the solution corresponding to the desired scalar product. This issue is possibly
related to the fact that Theorem \ref{funEQ} is not restricted to a particular representation of Bethe vectors entering the definition
of the scalar product $\mathcal{S}_n$. In addition linearity also implies that Eq. (\ref{FS}) is at most able to determine 
the desired scalar product up to an overall multiplicative factor independent of the variables in the set $\gen{\Lambda}^{1,n}$.

Here we are interested in on-shell scalar products of the $\mathcal{U}_q [\widehat{\alg{sl}}(2)]$ six-vertex model with twisted boundary
conditions \cite{deVega_1984}, and this choice fixes the representation of Bethe vectors entering the construction of $\mathcal{S}_n$.
The scalar product $\mathcal{S}_n$ for this particular model exhibits two important properties that will be required in our analysis.
These properties are formulated in the following lemmas.

\begin{lemma}[Symmetric function] \label{symmetric}
For $1 \leq i,j \leq n$ let $s_{ij}$ be the permutation operator acting on $f \in \mathbb{C}[\lambda_1^{\pm 1} , \lambda_2^{\pm 1} , \dots , \lambda_n^{\pm 1}]$ 
by 
\[
\left( s_{ij} f \right) (\dots, \lambda_i , \dots , \lambda_j , \dots ) \coloneqq f(\dots, \lambda_j , \dots , \lambda_i , \dots ) \; .
\]
Then $ s_{i,j} \mathcal{S}_n = \mathcal{S}_n$.
\end{lemma}

\begin{remark}
We remark that the symmetry property described in Lemma \ref{symmetric} has been already taken into account when writing the 
functional relation (\ref{FS}). It can be shown that, reversely, any analytic solution of (\ref{FS}) has this symmetry property.
\end{remark}

\begin{mydef} \label{Xset}
Let us introduce variables $x_i \coloneqq e^{2 \lambda_i}$ and set $\gen{X}^{i,j} \coloneqq \{ x_k \in \mathbb{C} \; | \; i \leq k \leq j \}$
and $\gen{X}^{i,j}_{x} \coloneqq \gen{X}^{i,j} \backslash \{ x \}$.
\end{mydef}

\begin{lemma}[Polynomial structure] \label{polynomial}
The function $\mathcal{S}_n$ is of the form
\[
\mathcal{S}_n ( \gen{\Lambda}^{1,n} ) = \frac{\bar{\mathcal{S}}_n ( \gen{X}^{1,n} )}{\displaystyle \prod_{x \in \gen{X}^{1,n}} x^{\frac{(L-1)}{2}}} \; ,
\]
where $\bar{\mathcal{S}_n}$ is a polynomial of degree $L-1$ in each of its variables separately.
\end{lemma}

The proofs of Lemmas \ref{symmetric} and \ref{polynomial} can be found in \cite{Galleas_SCP}. For completeness it is worth
remarking that, together with the asymptotic behavior of $\bar{\mathcal{S}}_n$ at infinity fixing the leading-term coefficient,
the above properties uniquely single out the on-shell scalar product among the solutions of \eqref{FS}.

\subsection{Partial differential equations}
\label{sec:PDE}

The possibility of expressing (\ref{FS}) as $\mathfrak{L}(\lambda_0) \mathcal{S}_n (\gen{\Lambda}^{1,n}) = 0$ with operator
$\mathfrak{L}$ given by (\ref{Lop}) is very suggestive. In fact, we have deliberately chosen to write formula (\ref{Lop}) as such in order to make
explicit the resemblance of $D_{\lambda}^{\lambda_0}$ with some sort of derivative. As we shall see, this is not out of reality when the action of
$D_{\lambda}^{\lambda_0}$ is restricted to a particular function space. 

\begin{proposition} \label{diff}
Let $\mathbb{C}[\gen{X}^{1,n}]$ denote the ring of polynomials in $n$ variables $x_1, x_2, \dots , x_n$ with complex coefficients.
Also, let $\mathbb{C}_m [\gen{X}^{1,n}] \subseteq \mathbb{C} [\gen{X}^{1,n}]$ be the subspace consisting of polynomials of degree at most $m$ in each variable $x_i$. Then 
\[
\label{real}
D_{x_i}^{x_0} = \sum_{k=0}^m \frac{(x_0 - x_i)^k}{k!} \frac{\partial^k }{\partial x_i^k}
\]
is a realization of (\ref{DIA}) on the space $\mathbb{C}_m [\gen{X}^{1,n}]$.
\end{proposition}
\begin{proof}
The proof follows readily from the Taylor series expansion of arbitrary functions in $\mathbb{C}_m [\gen{X}^{1,n}]$.
More details can be found in \cite{Galleas_2011, Galleas_proc}.
\end{proof} 

Although the realization (\ref{real}) exhibits an appealing structure, the solutions $\mathcal{S}_n$ we are interested do not immediately belong
to $\mathbb{C}_m [\gen{X}^{1,n}]$ as described in Lemma \ref{polynomial}. In fact, the desired on-shell scalar product belongs to the function
space $\mathcal{F}_n \coloneqq \mathbf{x}^{\frac{1-L}{2}} \mathbb{C}_{L-1}[\gen{X}^{1,n}]$ where $\mathbf{x} \coloneqq \prod_{\tilde{x} \in \gen{X}^{1,n}} \tilde{x}$.
Therefore we can see that the space $\mathcal{F}_n$ is not too far from $\mathbb{C}_{L-1}[\gen{X}^{1,n}]$. On the other hand, Lemma \ref{polynomial} also tells us
that the function $\bar{\mathcal{S}}_n$ actually lives in $\mathbb{C}_{L-1}[\gen{X}^{1,n}]$. In this way it is convenient to rewrite (\ref{FS}) in terms of functions
$\bar{\mathcal{S}}_n$ for which the realization (\ref{real}) is available. This motivates the following extra definitions.    

\begin{mydef}[Polynomial Bethe roots and inhomogeneity sets] \label{bethe_bar}
Write $y_i \coloneqq e^{2 \mu_i}$, $x_i^{B} \coloneqq e^{2 \lambda_i^B}$ and $q \coloneqq e^{\gamma}$. 
In addition to that, it is also convenient to define the sets $\bar{\mathcal{U}} \coloneqq \{ y_i \in \mathbb{C} \; | \; 1 \leq i \leq L \}$ and
$\bar{\mathcal{B}} \coloneqq \left\{ x_i^{B} \; | \; 1 \leq i \leq n \right\}$ such that
\< \label{BA_pol}
\prod_{y \in \bar{\mathcal{U}}} \frac{(x_i^B q - y q^{-1})}{(x_i^B  - y )} =  \frac{\phi_2}{\phi_1} \prod_{\substack{j=1 \\j \neq i}}^n \frac{(x_i^B q - x_j^B q^{-1})}{(x_i^B q^{-1} - x_j^B q)}  \; . 
\>
\end{mydef}

Taking into account the above discussion we then rewrite Eq. (\ref{FS}) as 
\[ \label{FSb}
\bar{\mathfrak{L}}(x_0) \; \bar{\mathcal{S}}_n (\gen{X}^{1,n}) = 0
\]
with operator $\bar{\mathfrak{L}}$ defined as
\[
\label{barLop}
\bar{\mathfrak{L}} (x_0) \coloneqq \bar{K}_0 + \sum_{x \in \gen{X}^{1,n}} \bar{K}_{x} \; D_{x}^{x_0} \; .
\]
The coefficients $\bar{K}_0$ and $\bar{K}_{x}$ in (\ref{barLop}) explicitly read
\<
\label{barcoeff}
\bar{K}_0 &\coloneqq& \phi_1 \prod_{y \in \bar{\mathcal{U}}} (x_0 q - y q^{-1}) \left[ \prod_{x \in \gen{X}^{1,n}} \frac{(x q - x_0 q^{-1})}{(x - x_0)} - \prod_{x \in \bar{\mathcal{B}}} \frac{(x q - x_0 q^{-1})}{(x - x_0)} \right] \nonumber \\
&& + \; \phi_2 \prod_{y \in \bar{\mathcal{U}}} (x_0  - y) \left[ \prod_{x \in \gen{X}^{1,n}} \frac{(x_0 q - x q^{-1})}{(x_0 - x)} - \prod_{x \in \bar{\mathcal{B}}} \frac{(x_0 q - x q^{-1})}{(x_0 - x)} \right] \nonumber \\
\bar{K}_x &\coloneqq& (q - q^{-1}) \frac{x_0}{(x_0 - x)}  \nonumber \\
&& \times \left[ \phi_1 \prod_{y \in \bar{\mathcal{U}}} (x q - y q^{-1}) \prod_{\tilde{x} \in \gen{X}_x^{1,n}} \frac{(\tilde{x} q - x q^{-1})}{(\tilde{x} - x)} - \phi_2 \prod_{y \in \bar{\mathcal{U}}} (x - y) \prod_{\tilde{x} \in \gen{X}_x^{1,n}} \frac{(x q - \tilde{x} q^{-1})}{(x - \tilde{x})} \right] \; . \nonumber \\
\>

The utility of this reformulation lies in the fact that $\bar{\mathfrak{L}}$, as defined by (\ref{barLop}), now acts on   
$\mathbb{C}_{L-1}[\gen{X}^{1,n}]$ and we can readily employ the differential realization (\ref{real}). Thus $\bar{\mathfrak{L}}$ can be 
regarded as a differential operator and its properties will be discussed in what follows.

\begin{theorem} \label{diff_struc}
The operator $\bar{\mathfrak{L}}$ is of the form 
\[
\label{Lstruc}
\bar{\mathfrak{L}}(x_0) = \frac{(q - q^{-1}) x_0}{\prod_{x \in \bar{\mathcal{B}}} (x_0 - x)} \sum_{k=0}^{L+n-2} x_0^k \; \gen{\Omega}_k \; ,
\]
where $\gen{\Omega}_k$ are differential operators that do not depend on $x_0$.
\end{theorem} 
\begin{proof}
We firstly substitute the realization (\ref{real}), adjusted to the present case, in (\ref{barLop}) and notice that the coefficients $\bar{K}_0$ and $\bar{K}_x$ only exhibits simple poles.
These poles are located at $x_0 = x_i$ and $x_0 = x_i^B$. It is not hard to see that the residues of $\bar{\mathfrak{L}}(x_0)$ at the points $x_0 = x_i$ vanish identically and, therefore, those points consist of removable singularities.
On the other hand, the residues of $\bar{\mathfrak{L}}(x_0)$ at the points $x_0 = x_i^B$ only vanish up to the Bethe ansatz condition (\ref{BA_pol}).
Thus we can conclude that $\bar{\mathfrak{L}}(x_0) \propto \prod_{x \in \bar{\mathcal{B}}} (x_0 - x)^{-1}$ and the remaining part
of (\ref{Lstruc}) is obtained from the explicit expansion of the coefficients $\bar{K}_0$ and $\bar{K}_x$ defined in (\ref{barcoeff}).
\end{proof}

The whole dependence of the equation (\ref{FSb}) with the variable $x_0$ is contained in the operator $\bar{\mathfrak{L}}$. As previously remarked, this property is a key ingredient to uncover a 
hierarchy of PDEs governing the scalar product $\bar{\mathcal{S}}_n$. 

\begin{corollary} \label{hierI}
The set $\{ \gen{\Omega}_k \mid 0 \leq k \leq L+n-2  \}$ forms a family of PDEs simultaneously integrated by $\bar{\mathcal{S}}_n$:
\[
\label{system}
\gen{\Omega}_k \; \bar{\mathcal{S}}_n (\gen{X}^{1,n}) = 0 \qquad \qquad 0 \leq k \leq L+n-2 \; .
\]
\end{corollary}
\begin{proof}
Consider Eq. (\ref{FSb}) taking into account Theorem \ref{diff_struc}. Since $\bar{\mathcal{S}}_n$ does not depend on $x_0$, we are thus left with (\ref{system}).
\end{proof}

\begin{remark}
The direct inspection of (\ref{system}) for small values of $L$ and $n$  indicates that any single equation might already able to
determine the desired on-shell polynomial solution up to an overall constant factor. This suggests that the set $\{ \gen{\Omega}_k \}_{k}$ might form a commutative family
of operators, i.e. $[ \gen{\Omega}_i, \gen{\Omega}_j ] = 0$. Here we shall not investigate this possibility but it certainly deserves further attention. 
\end{remark}

Although each differential operator $\gen{\Omega}_k$ can be written down for given values of $L$ and $n$, their explicit form can be rather
cumbersome. Fortunately, the leading-term operator $\gen{\Omega}_{L+n-2}$ exhibits an interesting structure and it can be compactly expressed for 
arbitrary values of $L$ and $n$. It turns out that the operator $\gen{\Omega}_{L+n-2}$ is explicitly given by
\<
\label{omega}
\gen{\Omega}_{L+n-2} = (\phi_1 q^{L+1-n} - \phi_2 q^{n-1}) \left( \sum_{\tilde{x} \in \bar{\mathcal{B}}} \tilde{x} - \sum_{x \in \gen{X}^{1,n}} x  \right) + \frac{1}{(L-1)!} \sum_{x \in \gen{X}^{1,n}} \mathcal{Y}_x \frac{\partial^{L-1}}{\partial x^{L-1}} \; , \nonumber \\
\>
with coefficients
\<
\label{YX}
\mathcal{Y}_x \coloneqq \phi_1 \prod_{y \in \bar{\mathcal{U}}} (x q - y q^{-1}) \prod_{\tilde{x} \in \gen{X}_x^{1,n}} \frac{(\tilde{x} q - x q^{-1})}{(\tilde{x} - x )} - \phi_2 \prod_{y \in \bar{\mathcal{U}}} (x - y ) \prod_{\tilde{x} \in \gen{X}_x^{1,n}} \frac{(x q - \tilde{x} q^{-1})}{(x - \tilde{x} )} \; . \nonumber \\
\>

The structure of $\gen{\Omega}_{L+n-2}$ as given by (\ref{omega}) and (\ref{YX}) exhibits interesting features on its own. Firstly, it is
similar to the ones found in \cite{Galleas_proc, Galleas_2015, Galleas_Lamers_2014} for six-vertex models with domain-wall boundaries. 
Moreover, analogously to the PDEs found in \cite{Galleas_proc, Galleas_2015, Galleas_Lamers_2014}, the scalar differential equation
$\gen{\Omega}_{L+n-2} \; \bar{\mathcal{S}}_n = 0$ can also be rewritten as a vector-valued first order PDE. We shall not perform this construction
here as it can be straightforwardly obtained from the analysis presented in \cite{Galleas_2015, Galleas_Lamers_2014}.   
Solutions of (\ref{omega}) for small values of $n$ will be discussed in what follows.

\subsubsection{Solutions and their properties}
\label{sec:sol}

The inspection of the equation $\gen{\Omega}_{L+n-2} ~ \bar{\mathcal{S}}_n (\gen{X}^{1,n}) = 0$ for small values of $L$ and $n$ reveals that
solutions $\bar{\mathcal{S}}_n \in \mathbb{C}_{L-1} [\gen{X}^{1,n}]$ only exist for parameters $x_i^{B}$ constrained by Bethe ansatz 
equations (\ref{BA_pol}). Our investigations seem to suggest that this single differential equation already determines
the on-shell scalar product for all $L$ and $n$. If this indeed the case we can readily see from (\ref{omega}) that the 
dependence of $\bar{\mathcal{S}}_n (\gen{X}^{1,n})$ with the parameters $x_i^{B}$ should only appear through the combination 
$\mathfrak{b} \coloneqq \sum_{\tilde{x} \in \bar{\mathcal{B}}} \tilde{x}$. We point out that this property is not apparent from 
Slavnov's determinant representation of on-shell scalar products \cite{Slavnov_1989}. Next we shall discuss the
integration of our PDE for some particular values of $n$.

\paragraph{Case $n=1$.} In that case let us first consider $L=2$ for which a general analysis of (\ref{omega}) is within reach. 
The operator (\ref{omega}) then produces the following ordinary differential equation,
\[ 
\label{n1L2}
\left[ \phi_1 \prod_{j=1}^2 (q x_1 - y_j q^{-1}) - \phi_2 \prod_{j=1}^2 (x_1 - y_j) \right] \frac{\dd \bar{\mathcal{S}}_1}{\dd x_1} + (\phi_1 q^2 - \phi_2) (x_1^B - x_1) \bar{\mathcal{S}}_1 = 0 \; ,
\]
whose general solution is given by
\[
\label{S1}
\bar{\mathcal{S}}_1 (x_1) = C \exp{\left\{  \frac{1}{2} \log{(\Lambda(x_1))} + \frac{\omega(x_1^B)}{\kappa} \arctanh{\left( \frac{\omega(x_1)}{\kappa} \right)} \right\} } \; . 
\]
In (\ref{S1}) $C$ stands for an integration constant while 
\[
\kappa \coloneqq \sqrt{ 2 \phi_1 \phi_2 \left[ q^2 (y_1 + y_2)^2 - 2 y_1 y_2 (1 + q^4) \right] - q^2 (\phi_1^2 + \phi_2^2)(y_1 - y_2)^2 } \; .
\]
The functions $\Lambda$ and $\omega$ are in their turn defined as
\<
\Lambda(x) &\coloneqq& \phi_1 \prod_{j=1}^2 (q^2 x - y_j) - \phi_2 q^2 \prod_{j=1}^2 (x - y_j) \nonumber \\
\omega(x) &\coloneqq& \ii q \left[ 2 x (\phi_2 - \phi_1 q^2) + (\phi_1 - \phi_2)(y_1 + y_2)  \right] \; .
\>

Now one can verify that the condition 
\[
\label{omk}
\left[ \frac{\omega(x_1^B)}{\kappa} \right]^2 = 1 
\] 
corresponds to the Bethe ansatz equation (\ref{BA_pol}). Moreover, assuming (\ref{omk}) and recalling the identity 
$\exp{\left[ \arctanh{(x)} \right]} = (1+x)^{\frac{1}{2}} (1-x)^{-\frac{1}{2}}$, one can check that the solution (\ref{S1}) 
collapses to the space $\mathbb{C}_{1} [x_1]$. 

Next we reverse our argument and consider the case $n=1$ and arbitrary $L$ starting from the assumption 
$\bar{\mathcal{S}}_1 \in \mathbb{C}_{L-1} [x_1]$. 
In that case $\frac{\dd^{L-1} \bar{\mathcal{S}}_1}{\dd x_1^{L-1}}$ is a constant and the integration of (\ref{omega})
is straightforward. We are thus left with
\[
\label{ss1}
\bar{\mathcal{S}}_1 \propto \frac{\displaystyle \left[ \phi_1 \prod_{y \in \bar{\mathcal{U}}} (x_1 q - y q^{-1}) - \phi_2 \prod_{y \in \bar{\mathcal{U}}} (x_1 - y )  \right]}{(x_1 - x_1^B)} \; .
\]
However, formula (\ref{ss1}) exhibits a simple pole when $x_1 \to x_1^B$. Thus, in order to have a polynomial, the residue of (\ref{ss1}) at this pole
should vanish. This requirement yields the following constraint on $x_1^B$,
\[
\prod_{y \in \bar{\mathcal{U}}} \frac{(x_i^B q - y q^{-1})}{(x_i^B  - y )} = \frac{\phi_2}{\phi_1} \; ,
\]
which can be immediately recognized as the Bethe ansatz equation (\ref{BA_pol}) for $n=1$. 

This analysis suggests that polynomial solutions and Bethe ansatz equations are equivalent requirements within our approach.

\paragraph{Case $n=2$.} The equation produced by (\ref{omega}) for $n=2$ and arbitrary $L$ assumes the form
\[
\label{n2L2}
Y_1 \frac{\partial^{L-1} \bar{\mathcal{S}}_2}{\partial x_1^{L-1}} + Y_2 \frac{\partial^{L-1} \bar{\mathcal{S}}_2}{\partial x_2^{L-1}} = V \bar{\mathcal{S}}_2 \; , 
\]
with coefficients explicitly defined as
\<
V &\coloneqq& ( \phi_1 q^{L-1} - \phi_2 q ) (x_1 + x_2 - \mathfrak{b}) \nonumber \\
Y_1 &\coloneqq& \frac{1}{(L-1)!} \left[ \phi_1 \frac{(x_2 q - x_1 q^{-1})}{(x_2 - x_1)} \prod_{y \in \bar{\mathcal{U}}} (x_1 q - y q^{-1}) - \phi_2 \frac{(x_1 q - x_2 q^{-1})}{(x_1 - x_2)} \prod_{y \in \bar{\mathcal{U}}} (x_1 - y)  \right] \nonumber \\
Y_2 &\coloneqq& \frac{1}{(L-1)!} \left[ \phi_1 \frac{(x_1 q - x_2 q^{-1})}{(x_1 - x_2)} \prod_{y \in \bar{\mathcal{U}}} (x_2 q - y q^{-1}) - \phi_2 \frac{(x_2 q - x_1 q^{-1})}{(x_2 - x_1)} \prod_{y \in \bar{\mathcal{U}}} (x_2 - y) \right] \; . \nonumber \\
\>
The requirement that $\bar{\mathcal{S}}_2 (x_1 , x_2)$ is a symmetric polynomial living in $\mathbb{C}_{L-1} [x_1 , x_2]$ implies 
a reduction of variables,
\[
\frac{ \partial^{L-1} \bar{\mathcal{S}}_2 (x_1 , x_2)}{\partial x_1^{L-1}} = H(x_2) \qquad \mbox{and} \qquad \frac{\partial^{L-1} \bar{\mathcal{S}}_2 (x_1 , x_2)}{\partial x_2^{L-1}} = H(x_1) \; , 
\]
where now $H(x) \in \mathbb{C}_{L-1} [x]$. Hence the resolution of (\ref{n2L2}) for $\bar{\mathcal{S}}_2$ yields the formula
\<
\label{SH}
\bar{\mathcal{S}}_2 (x_1 , x_2) &=& \frac{\Theta^{-1}}{(x_1 - x_2)} \frac{\left[ K(x_1, x_2) H(x_1) - K(x_2 , x_1) H(x_2) \right]}{(x_1 + x_2 - \mathfrak{b})} \; , 
\>
with $\Theta = (L-1)! \; (\phi_1 q^{L-1} - \phi_2 q)$ and
\[
K(x_1 , x_2) \coloneqq \phi_1 (x_1 q - x_2 q^{-1}) \prod_{y \in \bar{\mathcal{U}}} (x_2 q - y q^{-1}) + \phi_2 (x_2 q - x_1 q^{-1}) \prod_{y \in \bar{\mathcal{U}}} (x_2 - y) \; .
\]
Interestingly, expression (\ref{SH}) can be regarded as a sort of separation of variables induced by the requirement $\bar{\mathcal{S}}_2 (x_1 , x_2) \in \mathbb{C}_{L-1} [x_1 , x_2]$.
On the other hand, (\ref{SH}) appears to have singularities and, by symmetry, it suffices to focus on the variable $x_1$ to study this issue.  
Keeping in mind that $H(x)$ is a polynomial, we can see that the RHS of (\ref{SH}) only exhibits simple poles when $x_1 \to x_2$ and
$x_1 \to \mathfrak{b} - x_2$. As far as the pole $x_1 = x_2$ is concerned, the residue of (\ref{SH}) vanishes for arbitrary polynomials $H$,
so it is removable. By way of contrast, the residue of (\ref{SH}) at the pole $x_1 = \mathfrak{b} - x_2$ does 
not vanish identically. Thus, in order to have $\bar{\mathcal{S}}_2 (x_1 , x_2) \in \mathbb{C}_{L-1} [x_1 , x_2]$ we need to require the 
corresponding residue to vanish. This condition implies the following constraint on the function $H$,
\[
\label{KH}
K(\mathfrak{b} - x, x) H(\mathfrak{b} - x) - K(x, \mathfrak{b} - x) H(x) = 0 \; .
\]
Now, since $H(x) \in \mathbb{C}_{L-1} [x]$, we can write $H(x) = \sum_{i=0}^{L-1} c_i x^i$ and let (\ref{KH}) fix the coefficients $c_i$.
By doing so we immediately notice that this is only possible for values of $\mathfrak{b}$ constrained by a certain polynomial equation. 
The explicit form of this polynomial constraint for $L=2,3$ is given in what follows.

\begin{example} \label{EX}
For $L=2$  we find the constraint
\[ 
\label{jason}
\left[ q^2 (\phi_1 - \phi_2) \mathfrak{b}  + (y_1 + y_2)(\phi_2 q^2 - \phi_1) \right] \prod_{j=1}^2 \left[  q^2 \mathfrak{b}  - y_j (1 + q^2)  \right] = 0  \; , 
\]
while for $L=3$ it is given by
\[ 
\label{hugo}
\gen{\Upsilon} (\mathfrak{b})  \prod_{j=1}^3 \left[ q^2 \mathfrak{b}  - y_j (1 + q^2)  \right] = 0  \; . 
\]
In (\ref{hugo}) we have,
\[
\gen{\Upsilon} (\mathfrak{b}) \coloneqq q^3 (\phi_2 - \phi_1 q)^3 \mathfrak{b}^3  + 2 q^2 (\phi_2 - \phi_1 q)^2 (\phi_1 - \phi_2 q) \sum_{j=1}^3 y_j \mathfrak{b}^2  +  (\phi_2 - \phi_1 q) W \mathfrak{b} + W_0 \; , \nonumber \\
\] 
with parameters $W$ and $W_0$ defined as
\<
&& W \coloneqq q (\phi_1^2 + q^2 \phi_2^2) \left[ \sum_{j=1}^3 y_j^2 + 3 \sum_{1 \leq i < j \leq 3} y_i y_j   \right] \nonumber \\
&& \qquad \quad  - \; \phi_1 \phi_2 \left[ (1+q^4) \sum_{1 \leq i < j \leq 3} y_i y_j + 2 q^2 \left( \sum_{i=1}^3 y_i  \right)^2   \right] \nonumber \\
&& W_0 \coloneqq  (\phi_1 - \phi_2 q^3) \left\{ \phi_1 \phi_2 (q + q^{-1}) \left[ (q^2 - 4  + q^{-2}) y_1 y_2 y_3 -  \sum_{i=1}^3 y_i^2  \sum_{\substack{j=1 \\j \neq i}}^3 y_j   \right] \right. \nonumber \\
&& \qquad \qquad \qquad \qquad \quad + \left.  (\phi_1^2 + \phi_2^2) \prod_{1 \leq i < j \leq 3} (y_i +  y_j)  \right\} \; .
\>
\end{example}

It is interesting to notice that for both examples we find that the polynomial determining $\mathfrak{b}$ factorizes as
$P (\mathfrak{b})  \prod_{j=1}^L \left[ q^2 \mathfrak{b}  - y_j (1 + q^2)  \right] = 0$ where $P$ is a non-trivial polynomial. The trivial part
consists of a product of linear equations, where each equation depends on a single inhomogeneity $y_i$. Although those trivial terms are present in (\ref{jason})
and (\ref{hugo}), they might not correspond to scalar products of actual Bethe vectors. On the other hand, the polynomial equation
$P (\mathfrak{b}) =0$ seems to be the relevant one. In fact, one can notice its solutions yields the correct number of Bethe eigenvectors. 
For instance, for $n=2$ and $L=2$ we have only one solution in accordance with the number of admissible Bethe vectors. Similarly, for $n=2$ and $L=3$ we have
three solutions which is also the correct number of eigenvectors for the six-vertex model in that sector.

\section{Open boundary conditions}
\label{sec:open}

Vertex models of statistical mechanics can be defined with a variety of boundary conditions and a prominent role is played by
lattice systems with open (reflecting) boundary conditions. The extension of the QISM to this class of models
was put forward in \cite{Sklyanin_1988} wherein the so called \emph{reflection algebra} emerges as a fundamental structure governing
integrability at the boundaries. Prior to that, the exact diagonalization of certain spin-chain hamiltonians with open boundaries had been 
obtained in \cite{Gaudin_1971, Gaudin_book, Alcaraz_1987} by means of the celebrated Bethe ansatz. In fact, the case of open boundaries required
a generalization of the Bethe ansatz, which turns out to preserve the scattering theory interpretation. 

The inclusion of this type of boundary conditions within the framework of the QISM has several important consequences. For instance, this type
of boundary conditions have been shown to render systems fully invariant under the action of the underlying quantum affine algebra
\cite{Kulish_1991}. It is worth remarking here that the study of symmetries underlying a physical system is an important step of its analysis. 
As far as systems with open boundaries are concerned, the characterization of the model's symmetries has even found important applications within
the context of the AdS/CFT duality \cite{Hofman_2007, Galleas_2009}. In addition to that, the extension of the QISM to systems with open
boundaries paved the way for using a powerful algebraic machinery for studying six-vertex model correlation functions \cite{Kitanine_2007, Kitanine_2008}.

The evaluation of correlation functions within the framework of the QISM is intimately associated with the study of Bethe vectors scalar products. 
More precisely, the latter are building blocks of form factors \cite{Kitanine_1999, Kitanine_2007, Kitanine_2008} and this section is devoted
to the study of scalar products of the six-vertex model with open boundary conditions as solutions of certain PDEs. On the other hand it is 
worth remarking here that there exist alternative methods to study correlation functions for infinite lattices.
In particular, this study for open chains was presented in \cite{Kedem_1995a, Kedem_1995b} using the vertex operator approach and in 
\cite{Baseilhac_2013, Baseilhac_Kojima_2014a, Baseilhac_Kojima_2014b} by means of the $q$-Onsager algebra.

Here our analysis will be based on results recently presented in \cite{Galleas_openSCP} obtained via the algebraic-functional framework. As previously mentioned in \Secref{sec:funEQ}, here we are using 
the terminology \emph{on-shell scalar product} to refer to the case when only one Bethe vector is on shell. That is, the variables parameterizing
the Bethe vector are evaluated on solutions of the model's Bethe ansatz equations.
On-shell scalar products of the six-vertex model with open boundary conditions have been previously studied in \cite{Wang_2002, Kitanine_2007} 
using the method of \cite{Kitanine_1999}. In that case, a determinant representation resembling Slavnov's formula for the periodic case 
\cite{Slavnov_1989} has been obtained. Although an explicit formula for such a quantity is already available in the literature, our interest here 
lies in possible connections with the theory of PDEs and quantum many-body systems. 
In fact, similar connections have been observed for several quantities tackled through the algebraic-functional method.
For instance, this is the case for partition functions of six-vertex models with domain-wall boundary conditions 
\cite{Galleas_2011, Galleas_proc, Galleas_Lamers_2014} and the spectral problem of vertex models transfer matrices \cite{Galleas_2015}. 
In what follows we shall demonstrate that this is also the case for on-shell scalar products associated with the six-vertex model with 
diagonal open boundary conditions.

\subsection{Functional equation}
\label{sec:funEQ_open}

Among the main results of \cite{Galleas_openSCP} we have a set of two functional equations characterizing scalar products of Bethe vectors
for the six-vertex model with open boundary conditions. The simultaneous resolution of those functional equations, taking into account
certain properties expected for off-shell scalar products, yields the desired solution in terms of a multiple contour integral.
Similarly to the case with boundary twists considered in \cite{Galleas_SCP}, the on-shell case for open boundaries
can also be described by a single equation obtained as a linear combination of the equations presented in \cite{Galleas_openSCP}.
Here we shall make use of definitions and conventions introduced in \cite{Galleas_openSCP} since our present analysis builds on
results of the aforementioned work. Moreover, the open boundaries case will also require the following notation, in addition
to definitions \ref{defvar} and \ref{Uset} given in \Secref{sec:twist}.

\begin{mydef}[Double-row Bethe roots] \label{bethe_open}
Define the set $\mathcal{E} \coloneqq \{ \lambda_1^{\mathcal{B}} , \lambda_2^{\mathcal{B}} , \dots , \lambda_n^{\mathcal{B}} \}$
for $\lambda_k^{\mathcal{B}} \in \mathbb{C}$ solving the following system of Bethe ansatz equations for $1 \leq k \leq n$,
\<
\label{BA_open}
\frac{b(\lambda_k^{\mathcal{B}}+h)}{a(\lambda_k^{\mathcal{B}}-h)} \frac{b(\lambda_k^{\mathcal{B}}-\bar h)}{a(\lambda_k^{\mathcal{B}}+\bar h)} && \prod_{\mu \in \mathcal{U}} \frac{a(\lambda_k^{\mathcal{B}} - \mu)}{b(\lambda_k^{\mathcal{B}} - \mu)}\frac{a(\lambda_k^{\mathcal{B}} + \mu)}{b(\lambda_k^{\mathcal{B}} + \mu)} \nonumber \\ 
= && (-1)^{n-1} \prod_{\substack{i=1 \\ i \neq k}}^n \frac{a(\lambda_k^{\mathcal{B}} - \lambda_i^{\mathcal{B}})}{a(\lambda_i^{\mathcal{B}} - \lambda_k^{\mathcal{B}})} \frac{a(\lambda_k^{\mathcal{B}} + \lambda_i^{\mathcal{B}} +\gamma)}{b(\lambda_i^{\mathcal{B}} + \lambda_k^{\mathcal{B}})}   \; . 
\>
\end{mydef}

\begin{theorem}[Functional equation] \label{funEQ_open}
Let $\mathcal{T}_n \colon \mathbb{C}^n \to \mathbb{C}$ denote on-shell scalar products for the six-vertex model with open boundary conditions. Then
$\mathcal{T}_n$ satisfies the functional equation 
\[
\label{FS_open}
L_0 \; \mathcal{T}_n (\gen{\Lambda}^{1,n}) + \sum_{\lambda \in \gen{\Lambda}^{1,n}} L_{\lambda} \; \mathcal{T}_n (\gen{\Lambda}_{\lambda}^{0,n}) = 0 \; ,
\]
with coefficients $L_0$ and $L_{\lambda}$ defined as
\<
\label{coeff_open}
&& \begin{aligned}
L_0 \coloneqq \ & b(\lambda_0+h)b(\lambda_0-\bar h) \frac{a(2\lambda_0+\gamma)}{b(2\lambda_0+\gamma)} \prod_{\mu \in \mathcal{U}} a(\lambda_0 - \mu)a(\lambda_0 + \mu) \\ 
& \hphantom{+} \; \times \left[ \prod_{\lambda \in \gen{\Lambda}^{1,n}} \frac{a(\lambda - \lambda_0)}{b(\lambda - \lambda_0)}\frac{b(\lambda + \lambda_0)}{a(\lambda + \lambda_0)} -  \prod_{\lambda \in \mathcal{E}} \frac{a(\lambda - \lambda_0)}{b(\lambda - \lambda_0)}\frac{b(\lambda + \lambda_0)}{a(\lambda + \lambda_0)} \right]  \\
& + \; a(\lambda_0-h)a(\lambda_0+\bar h) \frac{b(2\lambda_0)}{a(2\lambda_0)} \prod_{\mu \in \mathcal{U}} b(\lambda_0 - \mu)b(\lambda_0 + \mu) \\ 
& \hphantom{+} \; \times \left[ \prod_{\lambda \in \gen{\Lambda}^{1,n}} \frac{a(\lambda_0 - \lambda)}{b(\lambda_0 - \lambda)}\frac{a(\lambda_0 + \lambda + \gamma)}{b(\lambda_0 + \lambda + \gamma)}  -  \prod_{\lambda \in \mathcal{E}} \frac{a(\lambda_0 - \lambda)}{b(\lambda_0 - \lambda)}\frac{a(\lambda_0 + \lambda + \gamma)}{b(\lambda_0 + \lambda + \gamma)}  \right] 
\end{aligned}  \nonumber \\ \\
&& \begin{aligned}
L_{\lambda} \coloneqq \ & \frac{a(2\lambda_0 +\gamma)}{a(\lambda_0 + \lambda)}\frac{c(\lambda_0 - \lambda)}{b(\lambda_0 - \lambda)} \frac{b(2\lambda)}{a(2\lambda)} \nonumber \\ 
& \times \left[ b(\lambda+h)b(\lambda-\bar h) \prod_{\mu \in \mathcal{U}} a(\lambda - \mu)a(\lambda + \mu) \prod_{\tilde{\lambda} \in \gen{\Lambda}^{1,n}_{\lambda}} \frac{a(\tilde{\lambda} - \lambda)}{b(\tilde{\lambda} - \lambda)}\frac{b(\tilde{\lambda} + \lambda)}{a(\tilde{\lambda} + \lambda)} \right. \nonumber \\
& \hphantom{\times} \; \left. - \; a(\lambda-h)a(\lambda+\bar h) \prod_{\mu \in \mathcal{U}} b(\lambda - \mu)b(\lambda + \mu) \prod_{\tilde{\lambda} \in \gen{\Lambda}^{1,n}_{\lambda}} \frac{a(\lambda - \tilde{\lambda})}{b(\lambda - \tilde{\lambda})}\frac{a(\lambda + \tilde{\lambda} + \gamma)}{b(\lambda + \tilde{\lambda} + \gamma)} \right] \; . \\
\end{aligned} 
\> 
\end{theorem}

\begin{proof}
Considering the results of \cite{Galleas_openSCP}, we multiply Eq. type A by $b(\bar{h} - \lambda_0) \frac{a(2 \lambda_0 + \gamma)}{b(2 \lambda_0 + \gamma)}$
and add it to Eq. type D multiplied by $a(\lambda_0 + \bar{h})$. Then we are left with (\ref{FS_open}) by assuming that the variables $\lambda_k^{\mathcal{B}}$ are constrained by
the Bethe ansatz equations (\ref{BA_open}).
\end{proof}

It is interesting to notice that (\ref{FS_open}) exhibits the same structure as (\ref{FS}) and the functional relations derived in 
\cite{Galleas_2013, Galleas_proc, Galleas_Lamers_2014} for partition functions with domain-wall boundaries. In this sense, both
(\ref{FS}) and (\ref{FS_open}) seem to be part of a general structure accommodating several quantities associated with integrable
vertex models.  Next we shall discuss the passage from (\ref{FS_open}) to a hierarchy of PDEs.

\subsubsection{Operatorial formulation}
\label{sec:OPD_open}

The procedure described here will go along the lines of that presented in \Secref{sec:OPD}. 
We shall start by reformulating the functional equation described in Theorem~\ref{funEQ_open} in terms of 
operators $D_{\lambda}^{\lambda_0}$ defined in (\ref{DIA}). This is the first step for converting (\ref{FS_open})
into a family of PDEs. In fact, this reformulation only depends on the structure of the equation and we can readily notice 
that (\ref{FS_open}) exhibits the same structure as (\ref{FS}). Therefore, we shall restrict ourselves to presenting only the final results. 

Let $\mathfrak{M}$ be an operator defined as
\[
\label{Mop}
\mathfrak{M}(\lambda_0) \coloneqq L_0 + \sum_{\lambda \in \gen{\Lambda}^{1,n}} L_{\lambda} \; D_{\lambda}^{\lambda_0} \; .
\]
In terms of $\mathfrak{M}$ the functional equation (\ref{FS_open}) can be written as $\mathfrak{M}(\lambda_0) ~ \mathcal{T}_n (\gen{\Lambda}^{1,n}) = 0$.
One important aspect of this reformulation is that the operators $D_{\lambda}^{\lambda_0}$ entering (\ref{Mop}) are defined in a rather abstract 
way \eqref{DIA}. This feature leaves room for choosing suitable realizations depending on the function space accommodating the desired
on-shell scalar product $\mathcal{T}_n$. This will be one of the aspects discussed in \Secref{sec:properties_open}.

\subsubsection{Properties of $\mathcal{T}_n$}
\label{sec:properties_open}

Most of the discussion of \Secref{sec:properties} applies to Eq. (\ref{FS_open}) and here we restrict ourselves to presenting
only special properties expected from the scalar product $\mathcal{T}_n$. Those properties are meant to guide us through the resolution 
of (\ref{FS_open}) ensuring that we have actually selected the solution corresponding to on-shell scalar products associated with the 
$\mathcal{U}_q [\widehat{\alg{sl}}(2)]$ six-vertex model with open boundaries \cite{Galleas_openSCP}.
In fact, the properties we shall use to single out actual on-shell scalar products characterize a function space; and for such function space
(\ref{FS_open}) gives rise to a hierarchy of PDEs. The required properties will be stated in lemmas \ref{symmetric_open} and
\ref{polynomial_open} whose proofs are available in \cite{Galleas_openSCP}.

\begin{lemma}[Symmetric function] \label{symmetric_open}
The function $\mathcal{T}_n$ is symmetric under the permutation of variables, i.e. $s_{ij} \mathcal{T}_n = \mathcal{T}_n$ for all $1 \leq i, j \leq n$.
\end{lemma}

\begin{remark}
As for the case with twisted boundary conditions, this symmetry property is already manifested in the functional relation \eqref{FS_open}.
\end{remark}

\begin{lemma}[Polynomial structure] \label{polynomial_open}
The function $\mathcal{T}_n$ is of the form
\[
\mathcal{T}_n ( \gen{\Lambda}^{1,n} ) = \frac{\bar{\mathcal{T}}_n ( \gen{X}^{1,n} )}{ \displaystyle \prod_{x \in \gen{X}^{1,n}} x^{L}} \; ,
\]
where $\bar{\mathcal{T}}_n \in \mathbb{C}_{2 L}[\gen{X}^{1,n}]$ is a multivariate polynomial of degree $2L$ in each variable separately.
\end{lemma}

\subsection{Partial differential equations}
\label{sec:PDE_open}

Solutions of $\mathfrak{M}(\lambda_0) ~ \mathcal{T}_n ( \gen{\Lambda}^{1,n} ) = 0$ corresponding to on-shell scalar products are 
characterized by Lemmas \ref{symmetric_open} and \ref{polynomial_open}. In other words, the desired scalar products consist of
solutions of our equation in the space $\mathcal{G}_n \coloneqq \mathbf{x}^{-L} \mathbb{C}_{2L}[\gen{X}^{1,n}]$.
The situation here is analogous to the case discussed in \Secref{sec:PDE} and a differential description of on-shell scalar products
for the case with open boundaries is available in terms of functions $\bar{\mathcal{T}}_n$. This is due to the fact that
$\bar{\mathcal{T}}_n \in \mathbb{C}_{2 L}[\gen{X}^{1,n}]$ for which (\ref{real}) can be employed.
The following additional definitions will also be useful. 

\begin{mydef}[Polynomial double-row Bethe roots] \label{bethe_bar_open}
Considering Definition \ref{bethe_bar} we further write  $t\coloneqq e^h$ and $\bar{t}\coloneqq e^{\bar{h}}$. We then define
$\bar{\mathcal{E}} \coloneqq \left\{ x_1^{\mathcal{B}}, x_2^{\mathcal{B}}, \dots , x_n^{\mathcal{B}} \right\}$ formed by variables
$x_k^{\mathcal{B}} = e^{2\lambda_k^{\mathcal{B}}} \in \mathbb{C}$ satisfying the polynomial form of (\ref{BA_open}) for $1 \leq k \leq n$, namely 
\< \label{BA_open_pol}
\begin{aligned}
\frac{\left(x_k^{\mathcal{B}} t-t^{-1} \right)}{\left(x_k^{\mathcal{B}} t^{-1} q-t q^{-1} \right)} \frac{\left(x_k^{\mathcal{B}} \bar{t}^{-1}-\bar{t}\right)}{\left(x_k^{\mathcal{B}} \bar{t} q-\bar{t}^{-1} q^{-1}\right)} & \prod_{y \in \bar{\mathcal{U}}} \frac{\left(x_k^{\mathcal{B}} q-y q^{-1}\right)}{\left(x_k^{\mathcal{B}}-y\right)} \frac{\left(x_k^{\mathcal{B}} q - y^{-1} q^{-1}\right)}{\left(x_k^{\mathcal{B}} - y^{-1} \right)} \\
& =  \prod_{\substack{i=1 \\ i \neq k}}^n \frac{\left(x_k^{\mathcal{B}} q - x_i^{\mathcal{B}} q^{-1}\right)}{\left(x_k^{\mathcal{B}} q^{-1} - x_i^{\mathcal{B}} q \right)} \frac{\left( x_k^{\mathcal{B}} q^2 - (x_i^{\mathcal{B}})^{-1}  q^{-2}\right)}{\left( x_k^{\mathcal{B}} - (x_i^{\mathcal{B}})^{-1} \right)}  \; .
\end{aligned}
\>
\end{mydef}

Next we rewrite (\ref{FS_open}) in terms of functions $\bar{\mathcal{T}}_n$ described in Lemma \ref{polynomial_open}. With the help
of Definition \ref{dia} we are then left with the equation $\bar{\mathfrak{M}} (x_0) ~ \bar{\mathcal{T}}_n (\gen{X}^{1,n}) = 0$
with operator $\bar{\mathfrak{M}}$ defined as
\[
\label{barLop_open}
\bar{\mathfrak{M}} (x_0) \coloneqq \bar{L}_0 + \sum_{x \in \gen{X}^{1,n}} \bar{L}_{x} \; D_{x}^{x_0} \; .
\]
The coefficients $\bar{L}_0$ and $\bar{L}_x$ in (\ref{barLop_open}) are in their turn given by
\[
\label{barcoeff_open}
\begin{aligned}
\bar{L}_0 \coloneqq \ & (x_0^{\frac{1}{2}} t - x_0^{-\frac{1}{2}} t^{-1})(x_0^{\frac{1}{2}} \bar{t}^{-1} - x_0^{- \frac{1}{2}} \bar{t}) \frac{\left( x_0 q^2 - x_0^{-1} q^{-2} \right)}{\left( x_0 q - x_0^{-1} q^{-1} \right)} \prod_{y\in\bar{\mathcal{U}}} (x_0 q-y q^{-1})(x_0 q - y^{-1} q^{-1}) \\
& \hphantom{+} \; \times \left[ \prod_{x \in \gen{X}^{1,n}} \frac{\left( x q - x_0 q^{-1} \right) }{\left( x - x_0 \right)} \frac{\left( x - x_0^{-1} \right)}{\left( x q - x_0^{-1} q^{-1} \right)} 
- \prod_{x \in \bar{\mathcal{E}}}  \frac{\left( x q - x_0 q^{-1} \right) }{\left( x - x_0 \right)} \frac{\left( x - x_0^{-1} \right)}{\left( x q - x_0^{-1} q^{-1} \right)}  \right] \\
& + \; (x_0^{\frac{1}{2}} q t^{-1} - x_0^{- \frac{1}{2}} q^{-1} t)(x_0^{\frac{1}{2}} q \bar{t} - x_0^{-\frac{1}{2}} q^{-1} \bar{t}^{-1}) \frac{\left( x_0 - x_0^{-1} \right)}{\left( x_0 q - x_0^{-1} q^{-1} \right)} \prod_{y\in\bar{\mathcal{U}}} (x_0 - y)(x_0  - y^{-1}) \\
& \hphantom{+} \; \times \left[ \prod_{x \in \gen{X}^{1,n}} \frac{\left( x_0 q - x q^{-1} \right)}{\left(x_0 - x \right)} \frac{\left( x_0  q^2 - x^{-1} q^{-2} \right)}{\left( x_0  q - x^{-1} q^{-1}\right)} 
- \prod_{x \in \bar{\mathcal{E}}} \frac{\left( x_0 q - x q^{-1} \right)}{\left(x_0 - x \right)} \frac{\left( x_0  q^2 - x^{-1} q^{-2} \right)}{\left( x_0  q - x^{-1} q^{-1}\right)}  \right] \\
\end{aligned} 
\]
\[
\begin{aligned}
\bar{L}_x \coloneqq \ & x_0 \frac{\left( q-q^{-1} \right)}{\left( x_0 - x \right)} \frac{\left( x_0 q^2 - x_0^{-1} q^{-2} \right)}{\left( x_0 q - x^{-1} q^{-1} \right)}  \frac{\left(x - x^{-1} \right)}{\left( x q - x^{-1} q^{-1} \right)} \\ 
& \times \left[ (x^{\frac{1}{2}} t - x^{-\frac{1}{2}} t^{-1})(x^{\frac{1}{2}} \bar{t}^{-1} - x^{-\frac{1}{2}} \bar{t}) \prod_{y\in\bar{\mathcal{U}}} (x q-y q^{-1})(x q - y^{-1} q^{-1}) \right. \\
& \qquad \qquad \qquad \qquad \qquad \qquad \qquad \qquad \times \left. \prod_{\tilde{x} \in \gen{X}^{1,n}_{x}} \frac{\left( \tilde{x} q - x q^{-1} \right)}{\left( \tilde{x} - x \right)} \frac{\left( \tilde{x} - x^{-1} \right)}{\left( \tilde{x} q - x^{-1} q^{-1} \right)} \right. \\
& \hphantom{\times} \; \left. - \; (x^{\frac{1}{2}} q t^{-1} - x^{-\frac{1}{2}} q^{-1} t)(x^{\frac{1}{2}} q \bar{t} - x^{-\frac{1}{2}} q^{-1} \bar{t}^{-1}) \prod_{y\in\bar{\mathcal{U}}} (x-y)(x  - y^{-1}) \right. \\
& \qquad \qquad \qquad \qquad \qquad \qquad \qquad \qquad \qquad \times \left. \prod_{\tilde{x} \in \gen{X}^{1,n}_{x}} \frac{\left( x q - \tilde{x} q^{-1} \right)}{\left( x - \tilde{x} \right)}\frac{\left( x  q^2 -  \tilde{x}^{-1} q^{-2} \right)}{\left( x  q - \tilde{x}^{-1} q^{-1} \right)}  \right] \; . 
\end{aligned} \nonumber 
\]

The operator $\bar{\mathfrak{M}}$ now acts on $\mathbb{C}_{2L} [\gen{X}^{1,n}]$ and we can
simply use the differential realization (\ref{DIA}) adapted to this case. Thus $\bar{\mathfrak{M}}$ can also be regarded as a differential
operator. Its structure is described in the following theorem.

\begin{theorem} \label{diff_struc_open}
The operator $\bar{\mathfrak{M}}$  is of the form 
\[
\label{Lstruc_open}
\bar{\mathfrak{M}} (x_0) = \frac{(q - q^{-1})}{\displaystyle \prod_{x \in \gen{X}^{1,n}} (x_0 q - x^{-1} q^{-1}) \prod_{x \in \bar{\mathcal{E}}} (x_0 q - x^{-1} q^{-1}) (x_0 - x)} \sum_{k=0}^{2L + 3n} x_0^k ~ \gen{\Phi}_k \; .
\]
In (\ref{Lstruc_open}) we use $\gen{\Phi}_k$ to denote differential operators independent of $x_0$.
\end{theorem} 
\begin{proof}
We proceed by inserting the differential realization (\ref{DIA}) into (\ref{barLop_open}). Then we notice from (\ref{barcoeff_open})
that $\bar{\mathfrak{M}}$, as a function of $x_0$,  only contains simple poles. These are located at $x_0 = \pm q^{-1}$, $x_0 \in \gen{X}^{1,n}$,
$x_0^{-1} \in q^2 ~\gen{X}^{1,n}$, $x_0 \in \bar{\mathcal{E}}$ and $x_0^{-1} \in q^2 ~\bar{\mathcal{E}}$. The residues at the poles
$x_0 = \pm q^{-1}$ and $x_0 \in \gen{X}^{1,n}$ vanish identically, hence they consist of removable singularities. On the other hand, the 
residues at the remaining poles do not vanish identically and give rise to the denominator in (\ref{Lstruc_open}). The numerator is now
a polynomial expression in $x_0$. The degree of this polynomial can be determined as follows. In order to compensate the dependence with $x_0$ in the denominator of \eqref{Lstruc_open} we need to include a
factor of $x_0^{3n}$ in the polynomial expansion. Next we count the powers coming from $\bar{L}_x D_x^{x_0}$ and notice that the leading term is proportional to $x_0^{2L}$.
Although the leading term coming from $\bar{L}_0$ seems to contain one more power of $x_0$, its careful inspection shows
that the coefficient of that term actually vanishes. Thus $\bar{L}_0$ also yields a leading term proportional to $x_0^{2L}$.
Thus we are left with a polynomial in $x_0$ of degree $2L+3n$.
\end{proof}

Our analysis will now proceed along the lines described in \Secref{sec:PDE}. For that the reformulation of (\ref{FS_open})
as $\bar{\mathfrak{M}} (x_0) \; \bar{\mathcal{T}}_n (\gen{X}^{1,n}) = 0$ plays a crucial role.

\begin{corollary} \label{hierII}
The set of differential operators $\{ \gen{\Phi}_k  \}$ constitutes a hierarchy of PDEs satisfying
\[
\label{system_open}
\gen{\Phi}_k  \; \bar{\mathcal{T}}_n (\gen{X}^{1,n}) = 0 \qquad \mbox{for} \quad 0 \leq k \leq 2L+3n \; .
\]
\end{corollary}
\begin{proof}
The proof is the same as for Corollary \ref{hierI}.
\end{proof}

Similarly to the case with twisted boundaries, the leading-term coefficient $\gen{\Phi}_{2L + 3n}$ in the expansion (\ref{Lstruc_open}) also plays a special role.
This coefficient can be written down explicitly for arbitrary values of $n$ and $L$ in a compact form. More precisely, it reads
\<
\label{phi}
\gen{\Phi}_{2L + 3n} &=& \left( t \bar{t}^{-1} q^{2L+1} - t^{-1} \bar{t} q^{4n-1} \right) \left[ \mathfrak{c} - \sum_{x \in \gen{X}^{1,n}} \left( x q + x^{-1} q^{-1} \right) \right]  \nonumber \\
&& + \; \frac{q^{2n + 1}}{(2L)!} \sum_{x \in \gen{X}^{1,n}} \frac{\left(x - x^{-1} \right)}{\left( x q - x^{-1} q^{-1} \right)}  \mathcal{Z}_x \frac{\partial^{2L}}{\partial x^{2L}} \; ,
\>
where the parameter $\mathfrak{c} \coloneqq \sum_{x \in \bar{\mathcal{E}}} \left( x q + x^{-1} q^{-1} \right)$ contains the whole dependence
of $\gen{\Phi}_{2L + 3n}$ with Bethe roots. In its turn the function $\mathcal{Z}_x$ has been defined as
\< \label{ZX}
\mathcal{Z}_x &\coloneqq& (x^{\frac{1}{2}} t - x^{-\frac{1}{2}} t^{-1})(x^{\frac{1}{2}} \bar{t}^{-1} - x^{-\frac{1}{2}} \bar{t}) \prod_{y\in\bar{\mathcal{U}}} (x q-y q^{-1})(x q - y^{-1} q^{-1})  \nonumber \\
&& \qquad \qquad \qquad \qquad \qquad \qquad \qquad \qquad \times  \prod_{\tilde{x} \in \gen{X}^{1,n}_{x}} \frac{\left( \tilde{x} q - x q^{-1} \right)}{\left( \tilde{x} - x \right)} \frac{\left( \tilde{x} - x^{-1} \right)}{\left( \tilde{x} q - x^{-1} q^{-1} \right)}  \nonumber \\
&&   - \; (x^{\frac{1}{2}} q t^{-1} - x^{-\frac{1}{2}} q^{-1} t)(x^{\frac{1}{2}} q \bar{t} - x^{-\frac{1}{2}} q^{-1} \bar{t}^{-1}) \prod_{y\in\bar{\mathcal{U}}} (x-y)(x  - y^{-1})  \nonumber \\
&& \qquad \qquad \qquad \qquad \qquad \qquad \qquad \qquad \qquad \times \prod_{\tilde{x} \in \gen{X}^{1,n}_{x}} \frac{\left( x q - \tilde{x} q^{-1} \right)}{\left( x - \tilde{x} \right)}\frac{\left( x  q^2 -  \tilde{x}^{-1} q^{-2} \right)}{\left( x  q - \tilde{x}^{-1} q^{-1} \right)} \; . \nonumber \\ 
\>

From (\ref{phi}) we can readily see that equation $\gen{\Phi}_{2L + 3n} ~\bar{\mathcal{T}}_n (\gen{X}^{1,n}) = 0$ is a linear PDE of order 
$2L$ in $n$ variables. Its structure is quite similar to (\ref{omega}) and the PDEs found in \cite{Galleas_2015, Galleas_Lamers_2014} describing
integrable vertex models. Although the order of (\ref{phi}) depends on the lattice length $L$, the corresponding differential equation can also
be recasted as a first-order vector-valued PDE along the lines of \cite{Galleas_2015}. Moreover, the fact that Bethe roots only enter (\ref{phi}) 
through the parameter $\mathfrak{c}$ will play an important role for its resolution. \Secref{sec:sol_open} will then be devoted to the analysis
of the equation $\gen{\Phi}_{2L + 3n} ~\bar{\mathcal{T}}_n (\gen{X}^{1,n}) = 0$ for small values of the magnon number $n$.

\subsubsection{Solutions and their properties}
\label{sec:sol_open}

Similarly to the case with boundary twists discussed in \Secref{sec:twist}, a determinant representation for $\mathcal{T}_n$ also exists
\cite{Wang_2002, Kitanine_2007}. However, it is again unclear from those results that  $\mathcal{T}_n$ might depend on Bethe roots
only through a particular combination. In this section we shall focus on the resolution of the single equation $\gen{\Phi}_{2L + 3n} ~\bar{\mathcal{T}}_n (\gen{X}^{1,n}) = 0$,
although Corollary \ref{hierII} entitles us with $2L+3n+1$ PDEs. A priori, it is not clear that a single equation would be able to
fix the desired solution but direct inspection of (\ref{system_open}) for small values of $L$ and $n$ suggests that each equation
separately is able to determine the desired on-shell polynomial solution up to an overall multiplicative factor. If this is indeed the
case then our results imply that the on-shell scalar products only depend on the Bethe roots through this particular combination
$\mathfrak{c}$. Moreover, we also find that solutions $\bar{\mathcal{T}}_n \in \mathbb{C}_{2L} [\gen{X}^{1,n}]$ 
only exists for parameters $x_i^{\mathcal{B}}$ constrained by the system of equations (\ref{BA_open_pol}). In what follows we shall demonstrate this fact explicitly for the cases $n=1,2$. 

\paragraph{Case $n=1$.} In this case our problem consists of solving the following ordinary differential equation,
\<
\label{tt1}
\frac{q^{3}}{(2L)!} \frac{\left(x_1 - x_1^{-1} \right)}{\left( x_1 q - x_1^{-1} q^{-1} \right)}  \mathcal{Z}_{x_1}  \frac{\dd^{2L} \bar{\mathcal{T}}_1}{\dd x_{1}^{2L}} + \left( t \bar{t}^{-1} q^{2L+1} - t^{-1} \bar{t} q^{3} \right) \left( \mathfrak{c} - x_1 q - x_1^{-1} q^{-1}  \right) \bar{\mathcal{T}}_1  = 0 \; . \nonumber \\
\> 
For solutions $\bar{\mathcal{T}}_1 \in \mathbb{C}_{2L} [x_1]$ we have that $\frac{\dd^{2L} \bar{\mathcal{T}}_1}{\dd x_{1}^{2L}}$
is a constant and we are left with
\[ 
\label{ttt1}
\bar{\mathcal{T}}_1 \propto \frac{\left(x_1 - x_1^{-1} \right)}{\left( x_1 q - x_1^{-1} q^{-1} \right)} \frac{\mathcal{Z}_{x_1}}{\left( x_1 q + x_1^{-1} q^{-1} - \mathfrak{c} \right)} \; .
\]
Formula (\ref{ZX}) for $n=1$ shows that $\mathcal{Z}_{x_1}$ contains no poles for finite values of the variable $x_1$.
Hence, the only source of poles in (\ref{ttt1}) is the denominator of (\ref{tt1}). More precisely, these poles are located at
\begin{enumerate}[label=\emph{\roman*}.] 
\item $(x_1 q)^2 = 1 \; ,$ 
\item $(x_1 q)^2 - x_1 q ~\mathfrak{c} + 1 = 0 \; . $
\end{enumerate}
The condition $\bar{\mathcal{T}}_1 \in \mathbb{C}_{2L} [x_1]$ now asks for the residues of (\ref{ttt1}) at those poles to vanish. 
As far as the poles (\textit{i}) are concerned, we can readily see that the corresponding residues vanish identically. However, the analysis
for the poles at (\textit{ii}) is slightly more involved. For that we first recall the definition of $\mathfrak{c}$ and notice that
(\textit{ii}) can be rewritten as $(x_1 - x_1^\mathcal{B}) (x_1 q^2 - 1/x_1^\mathcal{B}) = 0$. The residue of (\ref{ttt1}) at 
$x_1=x_1^\mathcal{B}$ vanishes provided $\mathcal{Z}_{x_1^\mathcal{B}} = 0$. This condition is equivalent to the Bethe ansatz equation \eqref{BA_open_pol} for $n=1$. 
Likewise, the pole at $x_1 = 1/(q^2 x_1^\mathcal{B})$ has zero residue when $x_1^\mathcal{B}$ is on shell. The latter follows from the fact
that the transformation $x_1^\mathcal{B} \mapsto 1/(x_1^\mathcal{B} q^2)$ only modifies \eqref{BA_open_pol} by an overall factor.
Hence, our analysis so far suggests that polynomial solutions and Bethe ansatz equations are equivalent requirements for the 
equation $\gen{\Phi}_{2L + 3n} ~\bar{\mathcal{T}}_n (\gen{X}^{1,n}) = 0$ also for the case of open boundaries.

\paragraph{Case $n=2$.} We proceed to the case $n=2$ along the same lines described in \Secref{sec:sol}. For that we assume $\bar{\mathcal{T}}_2 \in \mathbb{C}_{2 L}[x_1 , x_2]$
and from (\ref{phi}) one then finds
\[
\label{TH}
\bar{\mathcal{T}}_2 (x_1 , x_2) = \frac{\tilde{\Theta}^{-1}}{(x_1 - x_2)(x_1 x_2 q - q^{-1})} \frac{\left[ \tilde{K}(x_1, x_2) \tilde{H}(x_1) - \tilde{K}(x_2 , x_1) \tilde{H}(x_2) \right]}{\left[ (x_1 + x_2)(x_1 x_2 q + q^{-1}) - x_1 x_2 \; \mathfrak{c} \right]} \; ,
\]
where $\tilde{\Theta} \coloneqq (2L)! \; (t \bar{t}^{-1} q^{2L-4} - t^{-1} \bar{t} q^2)$. The function $\tilde{H}$ in (\ref{TH}) is a polynomial, more precisely 
$\tilde{H}(x) \in \mathbb{C}_{2L}[x]$, whilst $\tilde{K}$ is defined as
\<
&& \tilde{K}(x_1 , x_2) \coloneqq \frac{ x_1  \left(x_2 - x_2^{-1} \right)}{\left(x_2 q - x_2^{-1} q^{-1}\right)} \nonumber \\
&& \times \left[ (x_2 t - t^{-1})(x_2 \bar{t}^{-1} - \bar{t}) (x_1 q - x_2 q^{-1}) (x_1 x_2 - 1)  \prod_{y \in \bar{\mathcal{U}}} (x_2 q - y q^{-1})(x_2 q - y^{-1} q^{-1}) \right. \nonumber \\
&&  \qquad \qquad \qquad \quad  + \left. (x_2 q t^{-1} - q^{-1} t)(x_2 q \bar{t} - q^{-1} \bar{t}^{-1}) (x_2 q - x_1 q^{-1})(x_1 x_2 q^2 - q^{-2}) \right. \nonumber \\
&& \qquad \qquad \qquad \qquad \qquad \qquad \qquad \qquad \qquad \qquad \qquad \qquad  \times \left. \prod_{y \in \bar{\mathcal{U}}} (x_2 - y) (x_2 - y^{-1}) \right] \; . \nonumber \\
\>
In accordance with Lemma~\ref{symmetric_open}, formula (\ref{TH}) is manifestly symmetric under the permutation of variables $x_1$ and $x_2$.
This property allows us to focus on the variable $x_1$ in order to analyze the poles structure of (\ref{TH}). 
The function $\tilde{H}$, although not yet specified, does not contribute to this poles structure as it is a polynomial. Thus the RHS 
of \eqref{TH} only exhibits simple poles at
\begin{enumerate}[label=\emph{\roman*}.] 
\item $x_1=x_2 \; , $
\item $x_1=-1/(x_2 q^2) \; ,$
\item $(x_1 q)^2 = 1 \; ,$ 
\item $(x_1 + x_2)(x_1 x_2 q + q^{-1}) - x_1 x_2 \; \mathfrak{c} = 0 \; .$
\end{enumerate}
The residue at each of the poles (\textit{i})--(\textit{iii}) vanishes, so these singularities are removable.
By way of contrast, the poles at (\textit{iv}) require a more detailed analysis. 
In order to carry out this analysis we introduce parameterizations $u \coloneqq x_1 + x_2$ and $v \coloneqq x_1 x_2$. By doing so we avoid the appearance 
of square roots when looking at the algebraic curve defined by (\textit{iv}). The requirement that the residue of (\ref{TH}) vanish at (\textit{iv})
then yields the constraint
\[
\label{KtH}
\tilde{K}\big(x_1(u,v), x_2(u,v)\big) \tilde{H}\big(x_1(u,v)\big) - \tilde{K}\big(x_2(u,v), x_1(u,v)\big) \tilde{H}\big(x_2(u,v)\big) = 0
\]
for $u = v \mathfrak{c} (v q + q^{-1})^{-1}$. Next we write $\tilde{H}(x) = \sum_{i=0}^{2L} \tilde{c}_i x^i$ and let (\ref{KtH}) fix the coefficients $\tilde{c}_i$.
At this point it is important to emphasize that the use of variables $u$ and $v$ plays a fundamental role. They overcome the problem of dealing with square roots in the analysis
of (\ref{KtH}) allowing one to read out a linear system of equations for the coefficients $\tilde{c}_i$. Similar to the case with twisted boundaries discussed in 
\Secref{sec:PDE}, here our linear system also has more equations than coefficients $\tilde{c}_i$ to be fixed. The extra equations imposes a constraint on the parameter
$\mathfrak{c}$ that will be explicitly written down for $L=2,3$ in what follows.

\begin{example} 
For $L=2$ we find the following condition on the parameter $\mathfrak{c}$,
\[ \label{michael}
\begin{aligned}
& \left( \mathfrak{c} + \frac{(q^2 - q^{-2}) (1 + t^2 \bar{t}^2) y_1 y_2 + (q^2 \bar{t}^2 - q^{-2} t^2)	(y_1 + y_2) (1 + y_1 y_2)}{(q^{-1} t^2 - q \bar{t}^2) y_1 y_2} \right) \\
& \qquad\qquad\qquad\qquad\qquad\qquad\qquad\qquad \times \prod_{j=1}^2 \left(\mathfrak{c} - (q+q^{-1}) (y_1 + y_1^{-1})\right) = 0 \; .
\end{aligned}
\]
In order to avoid long formulae, for $L=3$ we restrict ourselves to the homogeneous case $y_i = 1$. In that case we find that $\mathfrak{c}$ has to satisfy
\[
\label{klaus}
\begin{aligned}
& \left( (t^2 - \bar{t}^2)^3 \; \mathfrak{c}^3 - (t^2 - \bar{t}^2)^2 \tilde{W}_2 \; \mathfrak{c}^2 + (t^2 - \bar{t}^2) \tilde{W}_1 \; \mathfrak{c} - \tilde{W}_0 \right) \; \prod_{j=1}^3 \left(\mathfrak{c} - 2(q+q^{-1})\right) = 0 \; ,
\end{aligned}
\]
with coefficients $\tilde{W}_0$, $\tilde{W}_1$ and $\tilde{W}_2$ explicitly reading
\[
\begin{aligned}
\tilde{W}_2 \coloneqq \ & (q - q^{-1}) (q^{-2} + 4 + q^2) (1 + t^2 \bar{t}^2) + 12 (q^{-1} t^2 - q \bar{t}^2) \; , \\
\tilde{W}_1 \coloneqq \ & (q - q^{-1})^2 (2 + q^2) (2 + q^{-2}) (1 + \bar{t}^4 t^4) \\
& + 6 (q - q^{-1}) 
\big[ (2 + q^2) (2 + q^{-2})(q^{-1} t^2-q \bar{t}^2) - (q^{-3} t^2 + q^3 \bar{t}^2) \big]
(1 + t^2 \bar{t}^2)  \; , \\
& - (q^3 - 51 q^{-1} + 2 q^{-3}) q^{-1} t^4 \; - \; (2 q^3 - 51 q^{-1} + q^{-3}) q \bar{t}^4 \\
& + (q^6 + 4 q^4 - 13 q^2 - 80 - 13 q^{-2} + 4 q^{-4} + q^{-6}) t^2 \bar{t}^2 \\ 
\tilde{W}_0 \coloneqq \ & (q - q^{-1}) (q^2 - q^{-2})^2 (1 + t^6 \bar{t}^6) \\
& + 6 (q - q^{-1}) (q^2 - q^{-2}) \big[(2 + q^{-2}) t^2 - (2 + q^2) \bar{t}^2\big] (1 + t^4 \bar{t}^4) \\
& - (q^2 - q^{-2}) (q^3 - 37 q - 13 q^{-1} + q^{-3}) q^{-2} t^4 (1 + t^2 \bar{t}^2) \\
& - (q^2 - q^{-2}) (q^3 - 13 q - 37 q^{-1} + q^{-3}) q^2 \bar{t}^4 (1 + t^2 \bar{t}^2) \\
& + (q^2 - q^{-2}) (q + q^{-1}) (q^4 + q^2 - 52 + q^{-2} + q^{-4}) t^2 \bar{t}^2 (1 + t^2 \bar{t}^2) \\
& + 2 (3 + q^{-2}) (q^2 - 12 + 3 q^{-2}) q^3 \bar{t}^6 \; + \ 2 (3 + q^2) (3 q^2 - 12 + q^{-2}) q^{-3} t^6 \\
& + 2 (16 q^5 - 36 q^3 - 57 q - 31 q^{-1} + 9 q^{-3} + 3 q^{-5}) t^4 \bar{t}^2 \\
& - 2 (3 q^5 + 9 q^3 - 31 q - 57 q^{-1} - 36 q^{-3} + 16 q^{-5}) t^2 \bar{t}^4  \; .
\end{aligned}
\]
\end{example}

\begin{remark}
The structure of the polynomial conditions (\ref{michael}) and (\ref{klaus}) is essentially the same as the one found for the parameter
$\mathfrak{b}$ in \Secref{sec:sol}. Thus one can see that the counting of admissible solutions $\mathfrak{c}$ is also in agreement with the number
of eigenvectors for the six-vertex model with open boundary conditions.
\end{remark}

\section{Concluding remarks}
\label{sec:CONCLUSION}

The evaluation of form factors for quantum integrable systems was shown in \cite{Slavnov_1989} to be directly
related to the computation of vertex models scalar products, and in this work we have presented a description of such
quantities in terms of linear PDEs. It is worth remarking that the description of certain correlation functions in terms of linear PDEs, 
namely KZ equations, is among the achievements of conformal symmetry \cite{Knizhnik_1984}. In this sense we can see
a convergence of those approaches at least at an ideological level. Although it is not clear if there exists a precise relation it would  
be interesting to further investigate this possibility.

In the present paper we have selected two types of vertex models to illustrate our method: the six-vertex model with generalized toroidal
boundary conditions, also known as twisted boundaries; and the six-vertex model with open boundary conditions. The twisted boundaries case
has been discussed in \Secref{sec:twist} while \Secref{sec:open} is concerned with the open boundaries case.
It is also important to stress here that our differential approach is one of the outcomes of the algebraic-functional method for scalar products proposed
in \cite{Galleas_SCP} and \cite{Galleas_openSCP}. Within this framework, the Yang-Baxter and reflection algebras plays a fundamental
role, and the PDEs obtained in \Secref{sec:PDE} and \Secref{sec:PDE_open} are a direct consequence of those algebraic structures.
This feature is also in consonance with the derivation of KZ equations by means of vertex algebras.
However, our approach not only renders a single PDE describing on-shell scalar products, but it produces a whole hierarchy of PDEs
for each one of the scalar products considered here.

Our PDEs also unveil interesting properties of on-shell scalar products for six-vertex models that were not apparent from the results previously
obtained in the literature \cite{Slavnov_1989, Wang_2002, Kitanine_2007}. From (\ref{omk}) and (\ref{phi}) we can see that the dependence of
$\mathcal{S}_n$ and $\mathcal{T}_n$ with the corresponding Bethe roots might only appear through particular linear combinations denoted respectively by
the parameters $\mathfrak{b}$ and $\mathfrak{c}$ in the main text. In \Secref{sec:PDE} we have shown that $\mathfrak{b}$ is determined by a single polynomial equation. 
This appears to represent a significant simplification since the standard approaches require the resolution of a system of Bethe ansatz equations (\ref{BA_pol}). 
In this way, our approach somehow seems to be able to combine the whole set of Bethe ansatz equations into a single equation for the 
quantity $\mathfrak{b}$, which is the one relevant for the on-shell scalar product $\mathcal{S}_n$. The same applies to the parameter $\mathfrak{c}$ 
in the case of $\mathcal{T}_n$ as shown in \Secref{sec:PDE_open}.

Nevertheless, there are still many questions that have eluded us so far. Foremost it would be important to prove rigorously that
the leading term differential operators $\Omega_{L+n-2}$ and $\Phi_{2L+3n}$ are indeed able to fix the desired scalar products.   
The particular combinations of Bethe roots forming the parameters $\mathfrak{b}$ and $\mathfrak{c}$ are quite appealing and one might wonder
if there is a more concrete meaning associated to them. Those parameters consist of simple sums of the associated Bethe roots, which suggests
they might correspond to the eigenvalue of one of the operators contained in the commuting family generated by the vertex model's transfer matrix.
Moreover, it is unclear to us at the moment how the system of Bethe ansatz equations could be suitably combined yielding the single equation
for the relevant parameter $\mathfrak{b}$ or $\mathfrak{c}$. Those questions certainly deserve further investigation.

\section{Acknowledgements}
\label{sec:ACK}

The authors thank G. Arutyunov for discussions and comments on this manuscript.
The work of W.G.\ is supported by the German Science Foundation (\textsc{dfg}) under the Collaborative Research Center 
(\textsc{sfb}) 676, \textit{Particles, Strings and the Early Universe}. The work of J.L.\ is supported by the Netherlands Organization 
for Scientific Research (\textsc{nwo}) under the \textsc{vici} grant 680-47-602 and by the \textsc{erc} Advanced Grant no.~246974, 
\textit{Supersymmetry: a window to non-perturbative physics}. J.L.\ also acknowledges the \textsc{d-itp} consortium, a program 
of \textsc{nwo} funded by the Dutch Ministry of Education, Culture and Science (\textsc{ocw}), and thanks \textsc{desy} for the hospitality
during the course of this work.

%
%
%

\bibliographystyle{hunsrt}
\bibliography{references}

\end{document}